\documentclass{article}
\usepackage[
  margin=3cm,
  includefoot,
  footskip=30pt,
]{geometry}
\RequirePackage[OT1]{fontenc}
\RequirePackage{amsthm,amsmath}
\RequirePackage[numbers]{natbib}

\usepackage{microtype}
\usepackage{graphicx}
\usepackage{subfigure}
\usepackage{algorithm,algorithmic}
\usepackage{amsmath}
\usepackage{amsthm}
\usepackage{bbm}
\usepackage{booktabs} 

\usepackage{hyperref}

\newtheorem{lemma}{Lemma}
\newtheorem{theorem}{Theorem}

\newcommand*\samethanks[1][\value{footnote}]{\footnotemark[#1]}

\title{Fair Correlation Clustering}
\usepackage{authblk}
\author[1]{Saba Ahmadi\thanks{The first two authors have contributed equally to this work.}}
\author[2]{Sainyam Galhotra\samethanks}
\author[3]{Barna Saha}
\author[4]{Schwartz Roy}
\affil[1]{Department of Computer Science,\protect\\
  University of Maryland,\protect\\
  \texttt{saba@cs.umd.edu}}
\affil[3]{College of Information and Computer Sciences,\protect\\
  University of Massachusetts Amherst,\protect\\
  Amherst, MA 01003\protect\\
  \texttt{sainyam@cs.umass.edu}}
\affil[3]{University of California Berkeley, USA\protect\\
        \texttt{barnas@berkeley.edu}}
\affil[4]{Technion - Israel Institute of Technology \protect\\
        \texttt{schwartz@cs.technion.ac.il}}

\begin{document}

\maketitle








\vskip 0.3in



\begin{abstract}


In this paper we study the problem of correlation clustering under fairness constraints. In the classic correlation clustering problem, we are given a complete graph where each edge is labeled positive or negative. The goal is to obtain a clustering of the vertices that minimizes disagreements –- the number of of negative edges trapped inside a cluster plus positive edges between different clusters.
We consider two variations of fairness constraint for the problem of correlation clustering where each node has a color, and the goal is to form clusters that do not over-represent vertices of any color. 

The first variant aims to generate clusters with minimum disagreements, where the distribution of a feature (e.g. gender) in each cluster is same as the global distribution. For the case of two colors when the desired ratio of the number of colors in each cluster is $1:p$, we get $\mathcal{O}(p^2)$-approximation algorithm.
Our algorithm could be extended to the case of multiple colors. 
We prove this problem is NP-hard.

The second variant considers relative upper and lower bounds on the number of nodes of any color in a cluster. The goal is to avoid violating upper and lower bounds corresponding to each color in each cluster while minimizing the total number of disagreements. 
Along with our theoretical results, we show the effectiveness of our algorithm to generate fair clusters by empirical evaluation on real world data sets.
\end{abstract}
\newpage
\section{Introduction}
The ubiquitous use of Machine learning tools for everyday decision making has brought the issue of fairness to the forefront. Many automated algorithms were shown to have implicit biases against certain demographies. In order to build machine learning algorithms that are inclusive, unbiased and helpful to the entire population, the recent years have seen a surge of research related to {\em fairness}. Many of the typical application scenarios where fairness has been identified to be crucial (e.g. lending, marketing, job selection etc.) requires clustering large datasets with sensitive features. These datasets often come as a network. In order to incorporate fairness into clustering, the seminal work by Chierichetti, Kumar, Lattanzi, and Vassilvitskii~\cite{chierichetti2017} proposed that each cluster has proportional representation from different demographic groups.They designed new approximation algorithms for the classic k-center and k-median clustering objectives with this notion of fairness. Subsequently, Schmidt, Schwiegelshohn, and Sohler~\cite{schmidt2018} extended the framework to k-means clustering. While these works study clustering algorithms over a metric space, many clustering applications work over network data, and that calls for designing {\em graph clustering} algorithms that are fair to all demographies. In a very recent work, Kleindessner, Samadi, Awasthi and Morgenstern consider the problem of fair spectral clustering~\cite{kleindessner2019guarantees}. They prove rigorous theoretical bounds for their algorithms over the stochastic block model. However, the analysis over arbitrary networks is still not known. In particular, the use of triangle inequalities makes the analysis of metric based fair clustering easier compared to graph clustering where the metricity is lacking.  

In this paper, we consider a fair variant of the classic optimization problem of correlation clustering. Correlation clustering is one of the most widely used clustering paradigms, and as claimed by Bonchi et al.~\cite{bonchi2014} ``arguably the most natural formulation of clustering''. Given a set of objects and a pairwise similarity measure between them, the objective is to partition the objects so that, to the best possible extent, similar objects are put in the same cluster and dissimilar objects are put in different clusters. This is represented by constructing a complete graph where edges are either labeled positive (similar objects) or negative (dissimilar objects). The edges can also be weighted. An algorithm for correlation clustering aims to minimize the disagreements among vertices, calculated as the weight of negative edges trapped inside a cluster plus positive edges between different clusters. As it just requires a definition of similarity, it can be applied broadly to a wide range of problems in different contexts such as social network analysis, data mining, computational biology, business and marketing \cite{veldt2018correlation,bonchi2014,hou2016new}.


Similar to other clustering algorithms, the known algorithms for correlation clustering may produce significantly biased output. In this work, we initiate a study of a fair variant of the correlation clustering problem where each vertex has a given feature, and the goal is to make sure that the distribution of the features is the same as the global distribution in each cluster. This is the same notion of fairness studied by Chierichetti et al.~\cite{chierichetti2017}, Bercea et al.~\cite{bercea2018cost} and Bera et al.~\cite{bera2019fair} on k-center and k-median. In another variation, the goal is to make sure the number of nodes of a specific feature $c_i$ in a cluster of size $n$ is between $\frac{n}{q_i}, \frac{n}{p_i}$ where $p_i\leq q_i\in \mathcal{Z}^{\geq 1}$, and $p_i, q_i$ are specified per each feature $c_i$. This later model was originally proposed by Ahmadian, Epasto, Kumar and Mahdian~\cite{ahmadian2019clustering} with only the upper bound and later Bera et al.~\cite{bera2019fair} generalized it to consider lower bound (both $p_i$ and $q_i$). Having a lower bound ensures that every color is represented in each cluster. They studied the $k$-center problem under this framework. In all our algorithms, we maintain the fairness constraints strictly, and optimize the objective function. That is, we give {\em exact} approximation algorithms, as opposed to {\em bi-criteria} approximation.

\subsection{Contributions \& Roadmap} 

Our contributions are as follows. 

For the first fairness variant, we can assume that each node has a color and the goal is to keep the distribution of the colors in each cluster same as the global distribution. First we show our results for the case of $2$ colors, and later we extend the results to an arbitrary number of colors.
Our approach for $2$ colors has some similarities to the approach proposed by Chierichitti et al.~\cite{chierichetti2017} for the k-center and k-median problems. The analysis of fair clusters with centroid based objectives leverages triangle inequality to bound the total objective value of the returned clusters. However, calculating the total disagreements for correlation clustering requires us to analyze the graph properties, leading to a completely different analysis as compared to \cite{chierichetti2017}. 

Assume nodes in the input graph are either red or blue and the goal is to have a ratio of $1:p$ of the number of red nodes to the number of blue nodes, where $p\in \mathcal{Z}^{\geq 1}$. 
In Section~\ref{sec:general-2-colors-CC}, we design a new algorithm with the following guarantees:
\begin{theorem}
\label{thm:2-colors-general}
Given a complete unweighted graph $G(V,E)$ where edges are labeled positive or negative, and the nodes are either red or blue where the ratio of number of red nodes to the number of blue nodes is $1:p$ for $p\in \mathcal{Z}^{\geq 1}$, there exists an algorithm which gives a clustering with ratio $1:p$ of number of red to blue nodes in each cluster and at most $\mathcal{O}(p^2)\cdot OPT$ disagreements.
\end{theorem}
Before explaining the general result for handling a ratio of $1:p$, in Section~\ref{sec:warm-up}, we explain a warm-up scenario where the desired ratio of red to blue in each cluster is $1:1$. In this section, the following theorem is proved:
\begin{theorem}
\label{thm:2-color-equal}
Given a complete unweighted graph $G(V,E)$ where edges are labeled positive or negative, and the nodes are either red or blue, with an equal number of red and blue nodes, there exists an algorithm which gives a clustering with equal number of red and blue nodes in each cluster and at most $(3\alpha+4)\cdot OPT$ disagreements, where $\alpha$ is the best approximation ratio for correlation clustering on a complete unweighted graph with minimizing disagreements objective. 
\end{theorem}

Our results could be generalized to multiple colors and in Section~\ref{sec:general-2-colors-CC}, a glimpse of the proof of the following theorem is provided. We delegate the complete proof to the Appendix.
\begin{theorem}
\label{thm:multiple-colors}
Given a complete unweighted graph $G(V,E)$ where edges are labeled positive or negative, and each node has exactly one of the colors $\{c_1,\cdots, c_{|\mathcal{C}|}\}$, and the ratio of the number of nodes of color $c_1$ to color $c_i$ is $1:p_i$ $(\forall 1 < i\leq |\mathcal{C}|)$, where $p_i \in \mathcal{Z}^{\geq 1}$, there exists an algorithm where the distribution of colors in each cluster is the same as the global distribution, and the total number of disagreemnts is at most $\mathcal{O}( (\max_{i=1}^{|\mathcal{C}|}\{p_i\})^2\cdot|\mathcal{C}|^2)\cdot OPT$.
\end{theorem}

In Section~\ref{sec:hardness}, we prove NP-hardness of fair correlation clustering problem on complete unweighted graphs even for $2$ colors. Note that the hardness result does not directly follow from the hardness result of the original correlation clustering.

In Section~\ref{sec:over-representation}, we consider the fairness model studied by Ahmadian et al.~\cite{ahmadian2019clustering} for $k$-center problem without over-representation. We are inspired by their definition of over-representation, and show the following theorem holds:
\begin{theorem}
\label{thm:over-representation}
Given a complete unweighted graph $G(V,E)$ where edges are labeled positive or negative, and nodes are colored red or blue, and two ratios $1:p, 1:q$ where $p,q\in \mathcal{Z}^{\geq 1}, p\leq q$, where ratio of the total number of red nodes to the total number of blue nodes is between $1:q$ and $1:p$, there exists an algorithm which gives a clustering where the ratio of number of red nodes to blue nodes in each cluster is between $1:q$ and $1:p$, and the total number of disagreements is at most $\mathcal{O}(q^2)\cdot OPT$.
\end{theorem}

We can extend Theorem~\ref{thm:over-representation} to the case with multiple colors.
\begin{theorem}
\label{thm:over-representation-multiple-colors}
Given a complete unweighted graph $G(V,E)$ where edges are labeled positive or negative, and each node has exactly one of the colors $\{c_1,\cdots,c_{\mathcal{C}}\}$, and two ratios $1:p_i, 1:q_i$ for each color $c_i$ where $p_i,q_i\in \mathcal{Z}^{\geq 1}, p_i\leq q_i$, where ratio of the total number of nodes of color $c_1$ to the total number of nodes of color $c_i$ needs to be between $1:q_i$ and $1:p_i$, there exists an algorithm which gives a clustering where $ \forall 1 < i \leq |\mathcal{C}|$, the ratio of number of nodes of color $c_1$ to color $c_i$ in each cluster is between $1:q_i$ and $1:p_i$, and the total number of disagreements is at most $\mathcal{O}((\max_{i=1}^{|\mathcal{C}|}\{q_i\})^2)\cdot OPT$.
\end{theorem}

In Section~\ref{sec:exp}, we perform an extensive evaluation on real world datasets to demonstrate the unfair results generated by the classical correlation clustering algorithm and  evaluate the ability of our algorithm to generate fair clusters without much loss of solution quality.


\section{Related Work}
Introduced by Bansal, Blum and Chawla in 2004~\cite{bansal2004correlation}, correlation clustering has received tremendous attention in the past decade. 
The problem is NP-complete, and a series of follow-up work have resulted in better approximation ratio, generalization to weighted graphs etc.~\cite{ailon2008aggregating, charikar2005clustering, chawla2015near}.
This problem captures a wide range of applications including clustering gene expression patterns~\cite{ben1999clustering,10.1007/978-3-540-79228-4_39}, and the aggregation of inconsistent information~\cite{filkov2004integrating}.

The research in fairness in machine learning has focused on two main directions, coming up with proper notions of fairness and designing fair algorithms. The first direction includes results on statistical parity~\cite{kamishima2012fairness}, disparate impact~\cite{feldman2015certifying}, and individual fairness~\cite{dwork2012fairness}. Second direction includes a bulk of work including fair rankings~\cite{celis2017ranking}, fair clusterings~\cite{chierichetti2017, rosner2018privacy, bercea2018cost, bera2019fair, ahmadian2019clustering}, fair voting~\cite{celis2017multiwinner}, and fair optimization with matroid constraints~\cite{chierichetti2019matroids}.

Puleo and Milencovic~\cite{puleo2018correlation} studied a new version of correlation clustering, where the objective was to make sure the maximum number of disagreements on each vertex is minimized. Their motivation was to make sure individuals are treated fairly. The result was improved by Charikar et al.~\cite{charikar2017local}. In a subsequent work, Ahmadi et al.~\cite{ahmadi2019min} studied the local correlation clustering problem where the objective was to make sure the maximum number of disagreements on each cluster is minimized, and the communities are treated fairly. Their result was improved by Kalhan et al.~\cite{kalhan2019improved}.

Chierichetti et al.~\cite{chierichetti2017} extended the notion of disparate impact to k-center and k-median, and studied these problems for the case of two groups. Their result was later generalized to multiple groups by R\"{o}sner and Schmidt~\cite{rosner2018privacy}. In this work, we generalize the notion of disparate impact to correlation clustering for multiple colors, and our goal is to make sure the distribution of colors in each cluster is identical to the global distribution. Next, we extend the model introduced by Ahmadian et al.~\cite{ahmadian2019clustering} on $k$-center to correlation clustering to show no color is over or under represented in each cluster.

\section{Preliminaries}
In the correlation clustering problem, an input graph $G(V,E)$ is given where each edge is labeled positive or negative. The goal is to obtain a clustering that minimizes the total number of disagreements, defined as the number of negative edges trapped inside a cluster plus positive edges that are cut between clusters.

Inspired by the recent developments on fairness in machine learning, we define a fair variant of correlation clustering problem. In fair correlation clustering problem, given an input graph $G(V,E)$ each node has a color from set of colors $\{c_1,\cdots, c_{|\mathcal{C}|}\}$. The desired ratio of the number of nodes from color $c_1$ to color $c_i$ is $1:p_i, \forall 1\leq i \leq |\mathcal{C}|$ where $\forall 1\leq i \leq |\mathcal{C}|:p_i \in \mathcal{Z}^{\geq 1}$. The goal is to find a clustering which gaurantees the desired ratios in each cluster while minimizing the total number of disagreements. First, we study fair correlation clustering problem for $2$ colors and then extend it to an arbitrary number of colors. 

In Section ~\ref{sec:over-representation}, we consider the problem of given an instance of correlation clustering where nodes are either red or blue and the ratio of the number of red nodes to blue nodes is between $1:q, 1:p$, where $p,q \in \mathcal{Z}^{\geq 1}, p\leq q$. The goal is to form a clustering where the ratio of the number of red nodes to blue nodes in each cluster is between $1:q$ and $1:p$. Throughout the paper, we use $OPT$ interchangeably for the optimum solution and minimum number of disagreements.
\section{Warmup: 2 Colors with Ratio 1:1}
\label{sec:warm-up}
\begin{algorithm}[tb]
   \caption{Fair Correlation Clustering }
   \label{alg:fairCC}
\begin{algorithmic}[1]
   \STATE {\bfseries Input:} $G(V=V_B\cup V_R,E=E^+\cup E^-)$
   \STATE $E'\leftarrow \phi$
   \FOR{$u\in V_B$}
   \FOR{$v\in V_R$}
    \STATE $E'\leftarrow E'\cup (u,v)$
    \STATE $w(u,v) \leftarrow  \sum\limits_{w\in V\setminus \{u,v\}}  \mathbbm{1}_{(u,w)\in E^+, (v,w)\in E^- } + \mathbbm{1}_{(u,w)\in E^-, (v,w)\in E^+ }$
   \ENDFOR
   \ENDFOR
   \STATE $\mathcal{C}'\leftarrow ClassicCorrClust(V_R,E\cap V_R\times V_R)$
   \STATE $M\leftarrow \texttt{min-weight-matching}(V,E',w)$
   \STATE $\forall v\in V_B$, assign $v$ to same cluster as $M(v)$
   \RETURN  $\mathcal{C}'$
   \end{algorithmic}
\end{algorithm}

 The goal is to find a clustering which minimizes the total number of disagreements, and the number of red and blue vertices in each cluster are equal. We show a constant approximation algorithm for this problem. 

Algorithm~\ref{alg:fairCC} presents our approach to generate clusters that obey the fairness constraint while minimizing the total disagreements. First, a weighted bipartite graph from $V_B$ to $V_R$ is constructed (lines 2-8) in the following way: consider a pair of vertices $(x,y)$ where $x\in V_B$ and $y \in V_R$, weight of this edge $w(x,y)$ is initially set to zero. If $(x,y)$ is a negative edge, increase $w(x,y)$ by $1$. For each vertex $z\in V\setminus \{x,y\}$, if the labels of edges $(z,x)$ and $(z,y)$ are different, increase $w(x,y)$ by $1$ (line 6). In this way, $w(x,y)$ shows how much the total disagreement increases if $x$ and $y$ are clustered together.  
Run an $\alpha$-approximation algorithm for minimizing disagreements on $V_R$ with no fairness constraints (line 9). Since $V_R\subset G$ and there are no fairness constraints on $V_R$, $OPT_{V_R}\leq OPT_G$.
Next, find a minimum weighted matching $M$ from $V_B$ to $V_R$ (line 10). In the end assign each blue vertex to the same cluster as its matched red node (line 11), and return the new clustering.
Let $w(M)$ denote weight of this matching. First, we show the following lemma holds:

\begin{lemma}
 $w(M)\leq 2\cdot OPT_G$.
\end{lemma}
\begin{proof}
Consider the optimal solution, and construct an arbitrary matching where both endpoints of each matched edge are in the same cluster. Call this matching $M'$ and assign weights to each matched edge as described in Algorithm~\ref{alg:fairCC}.  First, we show $w(M')\leq 2\cdot OPT_G$. Consider an edge $(v_i, v_j)$, if they are matched by the $M'$, $M'$ and $OPT$ are paying the same cost for this edge, since the matching is paying for this edge if and only if it is a negative edge trapped inside a cluster, in  this case $OPT$ is also paying for it.
Assume the case where $v_i, v_j$ are not matched in $M'$. Let's assume $v_i$ is matched to $v_{i'}$, and $v_j$ is matched to $v_{j'}$. The matching $M'$ could pay for the edge $(v_i, v_j)$ at most twice, once if the edges $(v_i, v_j)$ and $(v_{i'}, v_j)$ have disagreeing labels, and once if $(v_i, v_j)$ and $(v_i, v_{j'})$ have disagreeing labels. Therefore, $w(M')\leq 2\cdot OPT$. Therefore $w(M)\leq w(M')\leq 2\cdot OPT$ since we are finding a min cost matching $M$.
\end{proof}

Now consider the $\alpha$-approximation for $V_R$, and put each blue vertex in the same cluster as their matched red vertex (line 11). In the following we show this algorithm gives a $3\alpha+4$-approximation.

\subsection{Analysis:}
Consider vertices $x, x'\in V_R$ and $y, y'\in V_B$, where $(x,y)\in M$ and $(x', y')\in M$ (Figure~\ref{fig:warm-up}). In the following, we show that we can pay for all the disagreements within our $(3\alpha+4)\cdot OPT$ budget.


\begin{figure}
	\centering
	\subfigure[Warm-up example]{\includegraphics[width=0.45\columnwidth]{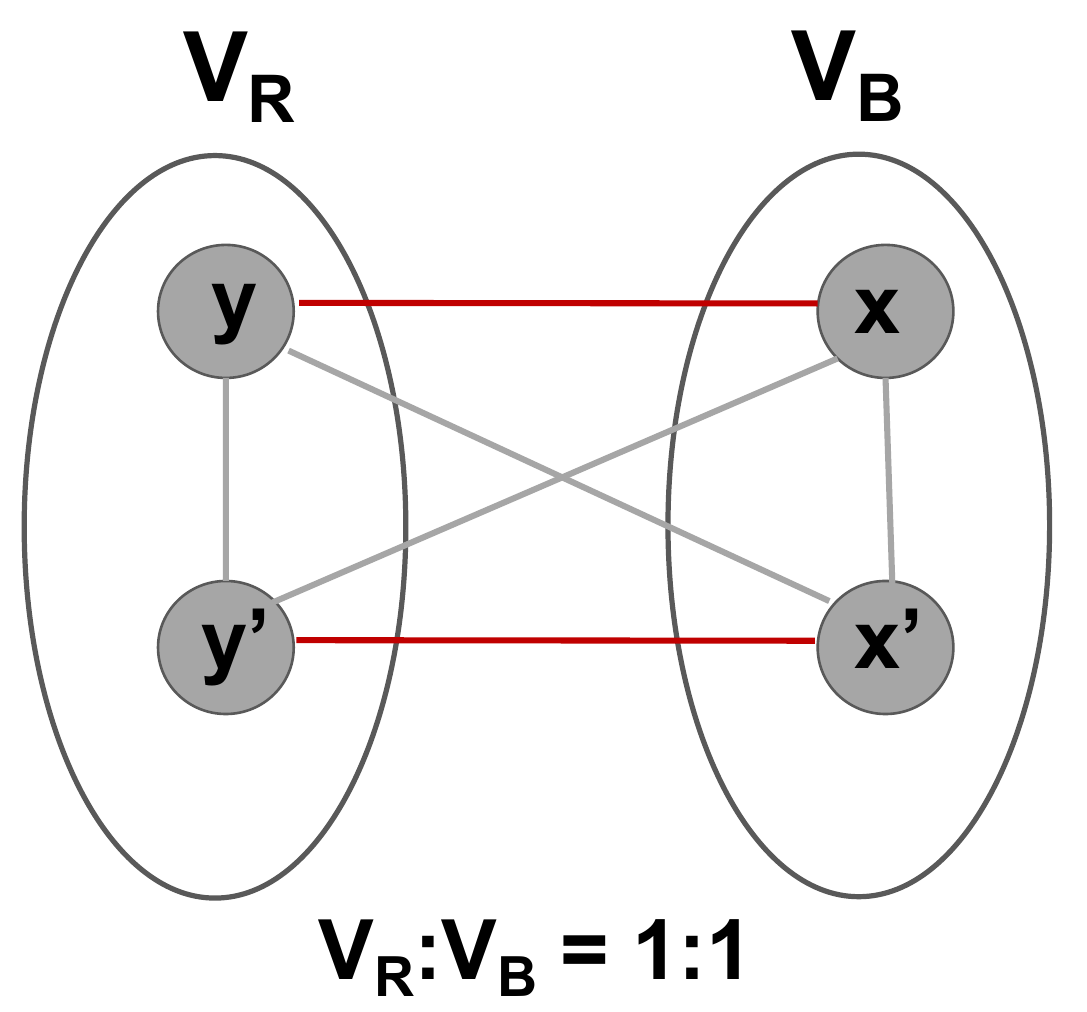}\label{fig:warm-up}}
	\subfigure[General case]{\includegraphics[width=0.45\columnwidth]{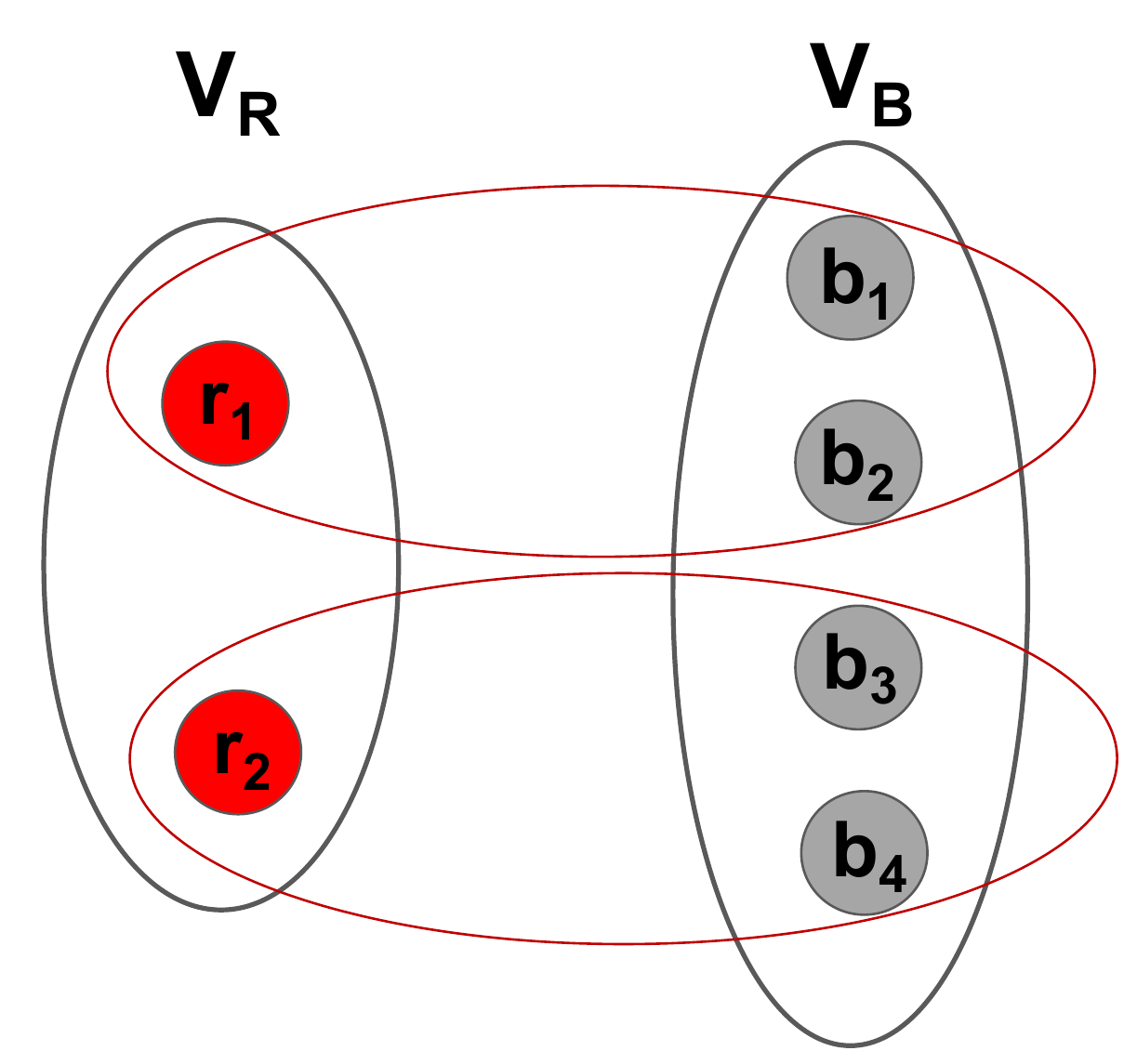}\label{fig:general-CC}}
	\vspace{-8pt}
	\caption{Example set of matched nodes for our algorithm.
	} 
	\label{fig:illus-1}
	\vspace{-14pt}
\end{figure}

\textbf{Case 1:} If a disagreement happens on a matched edge $(x,y)$, meaning $(x,y)$ is a negative edge, it is counted in $w(M)$.

\textbf{Case 2:} If a disagreement is on $(x', y)$, two cases may arise:
\begin{itemize}
\item{Case 2.1:} If $(x', y)$ and $(y', y)$ have the same label, since $x'$ and $y'$ are in the same cluster, having a disagreement on $(x', y)$ implies having a disagreemt on $(y', y)$. Therefore, if we double the budget needed to pay for the mistakes in $V_R$, we can also pay for the mistakes of this type.
\item {Case 2.2:} If $(x', y)$ and $(y', y)$ do not have the same label, we are making exactly one mistake on these two edges and the min cost matching $M$ is paying for it.
\end{itemize}

\textbf{Case 3:} If a disagreement occurs on $(x,x')$, the following cases might happen:
\begin{itemize}
    \item{Case 3.1:} If $(x,x')$, $(x',y)$ have different labels, we are making exactly one mistake on these two edges and the min cost matching $M$ is paying for it.
    \item{Case 3.2:} If $(x',x)$ and $(x',y)$ have the same labels, two cases might happen:
    \begin{itemize}
        \item{Case 3.2.1:} $(x',y)$ and $(y,y')$ have the same labels. In this case, $(x',x)$ and $(y',y)$ have the same labels, and $x'$, $y'$ are in the same cluster, also $x, y$ are in the same cluster. Therefore, there is a disagreement on $(x',x)$ if and only if there is a disagreement on $(y, y')$. By adding another $\alpha\cdot OPT$ to the budget, we can pay for these types of disagreements.
        \item {Case 3.2.2:} $(x', y)$ and $(y', y)$ have different labels. In this case, exactly one mistake occured on $(x',y)$ and $(y,y')$ and matching was paying for it. If no mistakes occures on $(x',y)$, there will be no mistake on $(x,x')$ as well. 
        If a disagreement happens on $(x',y)$, then a disagreement occurs on $(x',x)$. Since $M$ is paying for the disagreement occured on $(x',y)$, doubling the cost of $M$ in the budget pays for the mistake on $(x,x')$ as well.
    \end{itemize}
\end{itemize}

At the end, we get a $3\alpha+4$-approximation algorithm, which complete proof of Theorem~\ref{thm:2-color-equal}.
\section{Generalization}
In this section, we generalize the previously considered model to allow the ratio of colors to be $1:p$ where $p\in \mathcal{Z}^{\geq 1}$ in Section~\ref{sec:general-2-colors-CC} and allow more than 2 colors in Section~\ref{sec:multiple-colors}.
\subsection{2 Colors with Ratio 1:p} 

\label{sec:general-2-colors-CC}

The algorithm for this case is similar to Algorithm~\ref{alg:fairCC} with a minor difference; the matching $M$ constructed from $V_R$ to $V_B$ is a $b$-matching where the degrees of vertices in $V_R$ are $p$, and the degrees of vertices in $V_B$ are $1$. Let $w(M)$ denote weight of this matching. 
\begin{lemma}
\label{lemma:2colors-general}
$w(M)\leq 2p\cdot OPT$.
\end{lemma}
\begin{proof}
Consider the $OPT$ solution, and construct an arbitrary $b$-matching $M'$ from red nodes to blue nodes where degree of each red node is $p$, and degree of each blue node is $1$, and for each matched edge both its endpoints belong to the same cluster. Call this matching $M'$. First, we show $w(M')\leq 2p\cdot  OPT_G$.
Consider an edge between arbitrary vertices $v_i$ and $v_j$, such that they are not matched in $M'$. If a disagreement occurs on the edge between $(v_i, v_j)$ in $OPT_G$, this disagreement could have been counted at most $2p$ times in $w(M')$. Therefore $w(M')\leq 2p\cdot OPT_G$. Since $M$ is a min cost $b$-matching satisfying degree constraints:
$w(M)\leq w(M')\leq 2p\cdot OPT_G$
\end{proof}
The algorithm is as following: run an $\alpha$-approximation correlation clustering on a subset of $G$ which includes the red vertex from each hyper-node (i.e. a collection of matched nodes).
In the following, we show we can pay for all the disagreements within a $\Big((p^2+2p)\cdot\alpha + 4p^2\Big)\cdot OPT$ budget.

\textbf{Case 1:} In Figure~\ref{fig:general-CC}, consider a disagreement between a red vertex  ($r_1$) and a blue ($b_3$) node from different hyper-nodes. 
Two cases might happen:
\begin{itemize}
\item{Case 1.1:} If edges $(r_1, b_3)$ and $(r_1, r_2)$ have disagreeing labels, then cost of the edge $(r_2, b_3)$ counted in the matching  is paying for it.

\item{Case 1.2:} If edges $(r_1, b_3)$ and $(r_1, r_2)$ have the same signs, the disagreement on $(r_1, b_3)$ could be charged to the edge $(r_1, r_2)$. The number of such edges charged to $(r_1, r_2)$ is at most $2p$.
\end{itemize}
\textbf{Case 2:} There exists disagreement between two blue nodes from different hyper-nodes, like $(b_1, b_3)$ in Figure \ref{fig:general-CC}.
\begin{itemize}
    \item Case 2.1: Edges $(b_1, b_3)$ and $(r_1, b_3)$ are disagreeing. Then the cost of edge $(r_1, b_1)$ included in the cost of the matching  is paying for it.
    \item Case 2.2: Edges $(b_1, b_3)$ and $(r_1, b_3)$ have the same labels and have different labels with $(r_1, r_2)$. We charge the disagreement on $(b_1, b_3)$ to the edge $(r_2, b_3)$. There are $p$ choices for $b_1$, therefore at most $p$ edges of this type, plus the edge $(r_1, b_3)$ are charged to the edge  $(r_2,b_3)$, whereas $M$ is paying $1$ for the disagreement between $(r_1, r_2)$ and $(r_1, b_3)$. Therefore, we need to account for $p+1$ times the matching cost to account for all edges of this type.
    \item Case 2.3: Edges $(b_1, b_3), (r_1, b_3),(r_1, r_2)$ all have the same labels. There are $p^2$ choices for a pair of blue nodes like $(b_1, b_3)$, and disagreements on these edges could be charged to $(r_1, r_2)$.
\end{itemize}
\begin{algorithm}[tb]
   \caption{Fair Correlation Clustering for Multiple Colors}
   \label{alg:fairCC-multi-color}
\begin{algorithmic}[1]
   \STATE {\bfseries Input:} $G(V=V_{c_1}\cup V_{c_2}\cdots\cup V_{c_{\mathcal{C}}}, E=E^+\cup E^-)$
   \FOR{$1<i\leq |\mathcal{C}|$}
   \STATE $E'_i\leftarrow \phi$
   \FOR{$u\in V_{c_i}$}
   \FOR{$v\in V_{c_1}$}
    \STATE $E'_i\leftarrow E'_i\cup (u,v)$
    \STATE $w(u,v) \leftarrow  \sum\limits_{w\in V_i\setminus \{u,v\}}  \mathbbm{1}_{(u,w)\in E^+, (v,w)\in E^- } + \mathbbm{1}_{(u,w)\in E^-, (v,w)\in E^+ }$
   \ENDFOR
   \ENDFOR
   \STATE $M_i\leftarrow \texttt{min-weight-b-matching}(V_1\cup V_{c_i},E'_i,w)$
   \ENDFOR
   \STATE $\mathcal{C}'\leftarrow ClassicCorrClust(V_{c_1},E\cap V_{c_1}\times V_{c_1})$
   \FOR{$1<i\leq |\mathcal{C}|$}
    \STATE $\forall v\in V_{c_i}$, assign $v$ to same cluster as $M_i(v)$
   \ENDFOR
   \RETURN  $\mathcal{C}'$
   \end{algorithmic}
\end{algorithm}
\textbf{Case 3:} A disagreement between two blue nodes in the same hyper-node, $b_1$ and $b_2$ which means $(b_1,b_2)$ is a negative edge. If $(r_1,b_1)$ is positive then $(r_1,b_2)$'s contribution in the matching cost captures it. Similarly, if $(r_1,b_2)$ is a positive edge then the $(r_1,b_2)$'s contribution in matching cost captures this. If both $(r_1,b_1)$ and $(r_1,b_2)$ are negative edges then we can charge both the edges $1/2$. The total number of times an edge $(r_1,b_1)$ is charged is at most $p-1$ as there can be a maximum of $p-1$ negative edges from $b_1$.

There are a total of $p^2+2p$ charges on edges between red nodes (Cases 1.2 and 2.3) accounting for the total cost to be $(p^2+2p)C$, where $C$ is the correlation clustering objective on red vertices. Similarly, we charge each matched edge at most $p+1$ times their weight in Case 2.2 and at most $p-1$ times their weight in Case 3, the total contribution to the final objective is $2p\cdot w(M)$.  All charges required to handle cases 1.1 and 2.1 do not add any additional cost to the objective as they are already accounted for edges considered in $2p\cdot w(M)$.
Hence, the total objective value of returned clusters is at most:
\[(p^2+2p)C + (p+1)\cdot w(M) \leq \Big((p^2+2p)\cdot\alpha + 2p\times 2p\Big)\cdot OPT\]
Therefore the approximation ratio is $\mathcal{O}(p^2)$, and this completes proof of Theorem~\ref{thm:2-colors-general}.
\subsection{Multiple Colors}
\label{sec:multiple-colors}
Our results could be extended to the case of multiple colors. Assume there are $\mathcal{C}$ colors $\{c_1, c_2,\cdots, c_{|\mathcal{C}|}\}$, and the ratio of number of nodes of color $c_1$ to color $c_i$ is $1:p_i$ ($\forall 1 < i \leq |\mathcal{C}|$), where $p_i\in \mathcal{Z}^{\geq 1}$. In this case, we get an approximation ratio of $\mathcal{O}( (\max_{i=1}^{|\mathcal{C}|}\{p_i\})^2\cdot|\mathcal{C}|^2)$. Algorithmm~\ref{alg:fairCC-multi-color} is a generalization of Algorithm~\ref{alg:fairCC}. In this algorithm a set of $b$-matchings $\{M_i : 1< i \leq |\mathcal{C}|\}$ are constructed. Each matching $M_i$ is between nodes of color $c_1$ and  $c_i$ and degree of each node of color $c_1$ is $p_i$, and degree of each node of color $c_i$ is $1$.
Analysis of this algorithm is similar to the analysis of the algorithm for $2$ colors, and we delegate the complete proof to Appendix.

\subsection{Avoiding Over-representation}
\label{sec:over-representation}
In this section, we consider the model defined by Ahmadian et al.~\cite{ahmadian2019clustering} for $k$-center problem. Their goal is to make sure given a parameter $0\leq \alpha \leq 1$, maximum fraction of nodes in a cluster having a specific color is at most $\alpha$ times the size of the cluster. We consider the following problem: consider two colors red and blue. Our goal is to make sure the ratio of the number of red nodes to the number of blue nodes in each cluster is between $1:q$ and $1:p$ where $1\leq p \leq q$ and $p, q\in \mathcal{Z}^{\geq 1}$. The algorithm discussed in Section~\ref{sec:general-2-colors-CC} could be modified to handle this variation of the problem in the following way: when finding a minimum cost $b$-matching $M$, put the degree constraint on each red node to be between $p$ and $q$, and let the degree constraint on each blue node to be $1$. Therefore in the minimum cost matching $M$, each connected component has $1$ red node and at least $p$ and at most $q$ blue nodes. The rest of the algorithm is similar to the algorithm discussed in Section~\ref{sec:general-2-colors-CC}. Using a similar analysis, it could be seen the approximation ratio is $\mathcal{O}(q^2)$.
We show a sketch of the proof and discard full proof to the supplementary material.
\begin{proof}[Proofsketch]
Given the optimum solution $OPT$, we show a $b$-matching $M'$ could be constructed where endpoints of each matched edge belong to the same cluster, and the degree of each red node is at least $p$ and at most $q$. In the $OPT$ solution,  in each cluster, the ratio of the number of red to blue nodes is between $1:q$ and $1:p$. Consider a specific cluster $\mathcal{X}$ in the $OPT$ solution, let $n_r, n_b$ denote the number of red and blue nodes in this cluster respectively. Therefore, $n_r\cdot p \leq n_b \leq n_r\cdot q$. Construct a $b$-matching from red nodes to blue nodes in $\mathcal{X}$ as following:  assign $p$ distinct blue nodes to each red node. If any blue nodes are left un-assigned, assign them to any red node with matched degree less than $q$. Since $n_b \leq q\cdot n_r$, a $b$-matching satisfying degree constraints in each cluster of $OPT$ could be constructed. $M'$ is the union of the $b$-matchings formed in all the clusters.
Similar to the proof of Lemma~\ref{lemma:2colors-general}, we can show 
$w(M) \leq w(M')\leq 2q\cdot OPT$.
The rest of the proof is along the same lines as proof of Theorem~\ref{thm:2-colors-general}, and the algorithm outputs a clustering with at most $\mathcal{O}(q^2)\cdot OPT$ disagreements.
\end{proof}
In the case of multiple colors, where the goal is that in each cluster the ratio of the number of nodes of color $c_1$ to color $c_i$ be between $1:p_i$ and $1:q_i$ where $\forall 1<i\leq |\mathcal{C}|,  p_i \leq q_i$ and $p_i, q_i\in \mathcal{Z}^{\geq 1}$, Algorithm~\ref{alg:fairCC-multi-color} could be modified to handle this case: in each iteration, find a minimum cost weighted $b$-matching $M_i$ where the degree of each node of color $c_1$ is between $p_i$ and $q_i$, and the degree of each node of color $c_i$ is $1$. In the supplemetary material we show this algorithm obtains an approximation ratio of $\mathcal{O}( (\max_{i=1}^{|\mathcal{C}|}\{q_i\})^2\cdot|\mathcal{C}|^2)$ on the number of disagreements.
\subsection{Hardness}
\label{sec:hardness}
Consider a complete graph $G(V,E=E^+\cup E^-)$. Consider a new complete graph $H$ of $2|V|$ nodes which is contructed by duplicating the nodes of $G$ (say V and $V'=\{u'\mid u\in V\}$). We can assume the nodes of $ V$ to be colored blue and $V'$ can be considered as the mirror image of $V$ colored red. Each pair of nodes $u,v\in V$ are connected in the same way as $E$. A positive edge is added between $u$ and $u'$ for all $u\in V$, where $u'$ is the mirror image of $u$. For all $u\in V, v'\in V'$, the edge between $(u,v')$ has the same label as $(u,v)$ where $v$ is the mirror of $v'$. The graph $H$ restricted to vertices $V'$ is referred to as $G'(V',E')$ (as shown in Figure~\ref{fig:hardness}).

\begin{figure}[h]
\vspace{-14pt}
\includegraphics[scale=0.3]{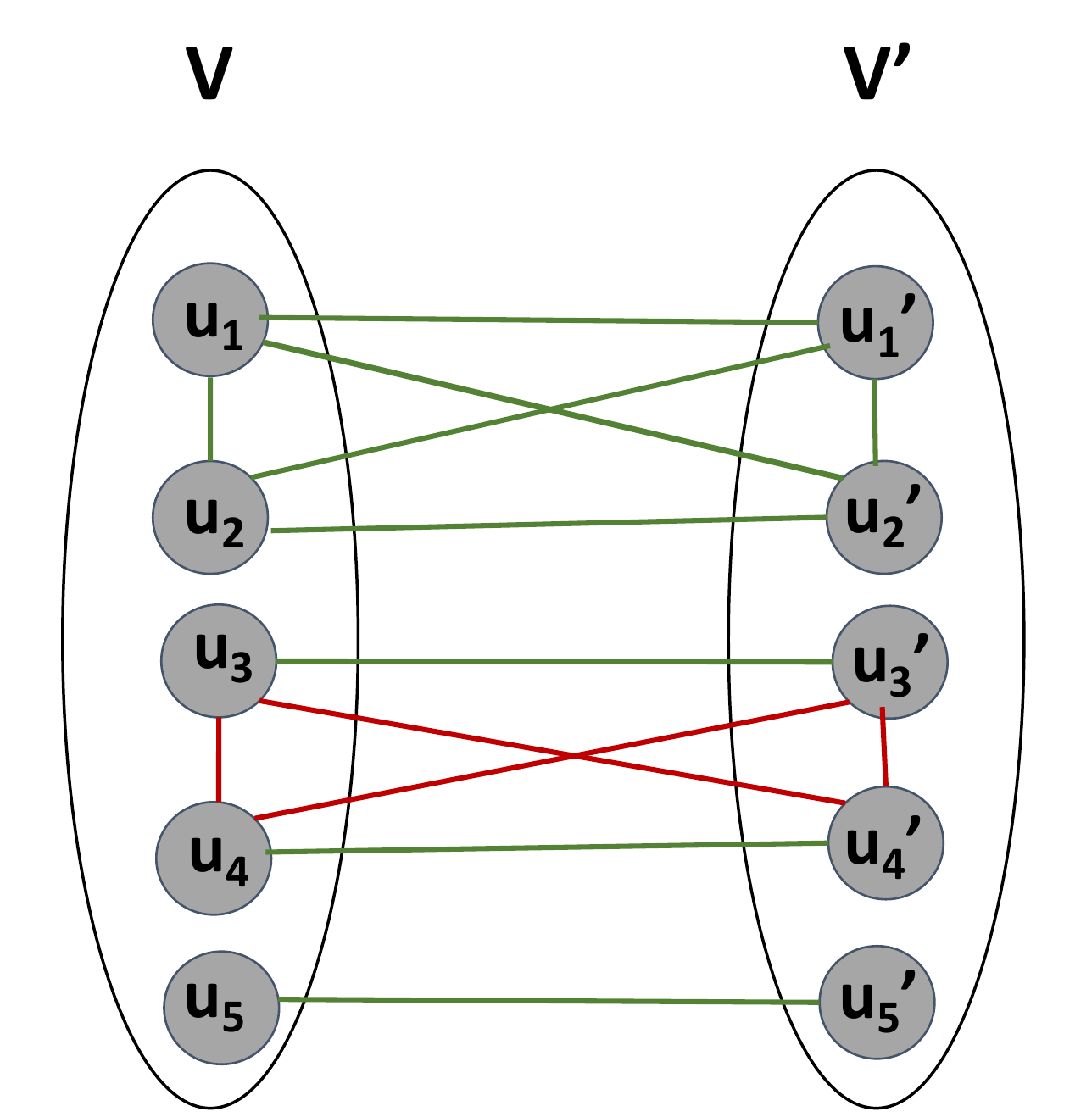}
\centering
\caption{Hardness example.}
\label{fig:hardness}
\end{figure}
Consider a clustering of $H$ with equal number of red and blue vertices in each cluster (say $\mathcal{C}'$). Now, we calculate the disagreements on the edges between nodes of $V$ and $V'$ to bound the total disagreements of $\mathcal{C}'$. A disagreement edge $(u',v')\in E'$ will lead to the following scenarios.
\begin{itemize}
\vspace{-8pt}
    \item{Case 1:} If $(u',v')\in E^-$: Therefore nodes $u'$ and $v'$ belong to same cluster.
    \begin{itemize}
    \vspace{-2pt}
        \item{Case 1.1:} If $u,v$ belong to the same cluster as $u'$ and $v'$, then edges $(u,v')$ and $(u',v)$ are mistakes.
        \item{Case 1.2:} $u$ belongs to same cluster as $u'$ and $v'$ but $v$ belongs to a different cluster. Edges $(v,v')$ and $(u,v')$ are the mistakes.
        \item{Case 1.3:} $u$ and $v$ belong to different cluster from $u'$ and $v'$, edges $(u,u')$ and $(v,v')$ are  mistakes
    \end{itemize}
    \item{Case 2: } If $(u',v')$ is a positive edge. This means $u'$ and $v'$ belong to different clusters.
    \begin{itemize}
    \vspace{-2pt}
        \item{Case 2.1:} If $u$ belongs to same cluster as $u'$ and v belongs to same cluster as $v'$ then edges $(u,v')$ and $(u',v)$ are the mistakes.
        \item{Case 2.2:} If $u$ belongs to different cluster from $u'$ and $v$ belongs to different cluster from $v'$ then $(u,u')$ and $(v,v')$ are the mistakes.
        \item{Case 2.3:} If $u$ belongs to different cluster from $u'$ and $v, v'$ belong to the same cluster:
        \begin{itemize}
            \item{Case 2.3.1:} If $u$ and $v'$ belong to different clusters then $(u,v')$ and $(u,u')$ are  mistakes.
            \item{Case 2.3.2:} If $u$ and $v'$ belong to the same cluster then $(u,u')$ and $(u',v)$ are mistakes.
            \vspace{-8pt}
        \end{itemize}
    \end{itemize}
\end{itemize}
This shows that for every disagreement $(u',v')\in E'$, there exist at least 2 disagreements in the edges between $\{u',v'\}$ and $\{u,v\}$:  $(u,v'),(u',v),(u,u')$ and $(v,v')$. If the disagreements on the subgraph of $H$ limited to $(V',E')$ is $O_{G'}$, then the total disagreements on the edges between $V$ and $V'$ is at least $2O_{G'}$. Hence, the total disagreements of $\mathcal{C}'$ is atleast $3O_{G'}+O_G$, where $O_G$ is the disagreements on subgraph limited to $G(V,E)$. Symmetrically performing the above mentioned analysis on $(u,v)\in E$, the total disagreements of $\mathcal{C}'$ is atleast  $3O_{G}+O_{G'}$. Hence the total disagreements of $\mathcal{C}'$ is at  least $\max\{3O_{G'} + O_{G}, 3O_G + O_{G'}\}$, which is minimized when $O_G=O_{G'}=OPT_{G}$. This is minimized when each red node and its mirror image belong to the same cluster. By discarding the nodes of $V'$ from the optimal solution of $H$, we get the optimal solution of classical correlation clustering on $G$. 

\section{Experiments\label{sec:exp}}
In this section, we empirically evaluate our algorithm along with some baselines on real world datasets. We show that the clusters generated by classical correlation clustering algorithm are unfair and our algorithm returns fair clusters without much loss in the quality of the clusters.

\noindent \textbf{Datasets.} We consider the following datasets.

\noindent 
\textit{Bank}\footnote{\url{https://archive.ics.uci.edu/ml/datasets/Bank+Marketing}}. This dataset comprises of phone call records of a marketing campaign run by a Portuguese bank. The \emph{marital} status of the clients is considered feature to ensure fairness.

\noindent \textit{Adult}\footnote{\url{https://archive.ics.uci.edu/ml/datasets/adult}}. 
Each record in the dataset represents a US citizen whose information was collected during 1994 census. 
We consider the feature \emph{sex} for fairness.

\noindent \textit{Medical Expenditure}\footnote{\url{https://meps.ahrq.gov/mepsweb/}}. The dataset contains medical information of various patients collected for research purposes. The \emph{race} attribute is considered for fairness.

\noindent \textit{Compas}\footnote{\url{https://github.com/propublica/compas-analysis}}. The dataset comprises of records of criminal trials used to analyze criminal recidivism. We consider race attribute for fairness.


Each record in the above mentioned datasets are considered as the nodes of the graph and the edge sign is determined by attribute simialrity between the nodes. We consider a sample of 1000 nodes in the above datasets for our experiments.
 \begin{figure*}[ht!]
\includegraphics[width=\columnwidth]{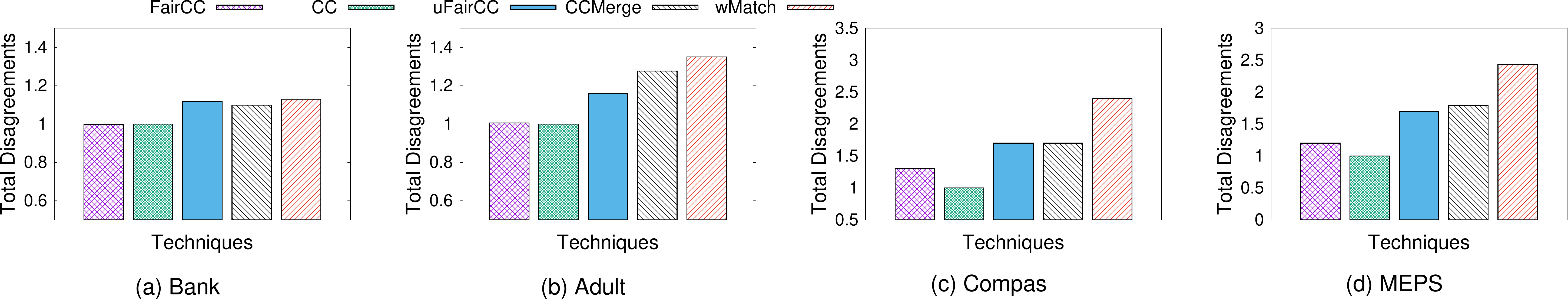}
\caption{Comparison of total disagreements for the different baselines with a constraint of ratio of two colors to be between 1:1 and 1:2.\label{fig:quality}}
\end{figure*}
 
\begin{figure*}[ht!]
\includegraphics[width=\columnwidth]{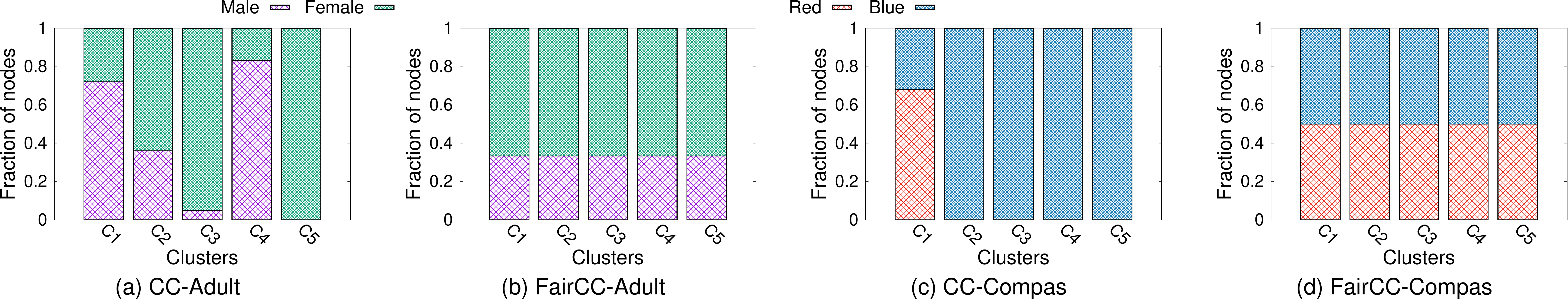}
\caption{Distribution of nodes of different colors in top 5 clusters generated by the algorithms.\label{fig:distribution}}
\end{figure*}

\noindent \textbf{Baselines.} 
We compare the quality of clusters returned by our algorithm with the following baselines focussed towards minimizing disagreements and ensure fairness. (i) \texttt{CC} -- The classical correlation clustering algorithm~\cite{ailon2008aggregating} that guarantees a 3-approximation of the optimal solution but does not ensure fairness  (iii) \texttt{wMatch} -- It generates a matching between nodes of different color as discussed in Algorithm~\ref{alg:fairCC} (iv) \texttt{uFairCC} -- Same as our algorithm with a difference that the matching component considers unit weight on inter-color edge.  (v) \texttt{CCMerge} -- This algorithm runs classical correlation clustering algorithm to generate initial clusters and then greedily add nodes to the clusters in decreasing size, so as to ensure fairness constraints. 

All the algorithms were implemented by us in Python using the networkx library on a 64GB RAM server. We run each algorithm 5 times and report average results. We calculate the total disagreements of the returned clusters to evaluate their quality. We denote our algorithm by \texttt{FairCC}. 

\subsection{Solution Quality}
This section compares the quality of clusters returned by the different algorithms for the specified distribution of features in each cluster.

\noindent \textbf{Fair proportion}
Figure~\ref{fig:quality} compares the total disagreements of the clusters returned by different algorithms. We observe similar trends across all the datasets. The clusters returned by \texttt{CC} do not obey the fairness constraint but all the other techniques ensure fairness. Across all datasets, \texttt{FairCC} achieves the minimum value of total disagreements as compared to the baselines that ensure fairness. Additionally, the loss in quality of clusters to achieve fairness as compared to \texttt{CC} is quite low. The matching returned by \texttt{wMatch} is same as that of \texttt{FairCC} but it achieves poor quality due to a number of positive edges going across the different components. The \texttt{CCMerge} algorithm ends up merging nodes which are connected by negative edges to ensure fairness, thereby losing on quality. \texttt{uFairCC} is same as our proposed solution except that the matching component between nodes of different colors does not consider weights. Superior performance of \texttt{FairCC} as compared to \texttt{uFairCC} justifies the benefit of our construction of a weighted bipartite graph to match nodes of different colors.

Figure~\ref{fig:distribution} shows distribution of top-5 clusters generated by \texttt{CC} and \texttt{FairCC} on Adult and Compas. The skew in distribution of the nodes of two colors in the clusters demonstrates the extent of unfairness in the results generated by classical correlation clustering algorithm. On the other hand, \texttt{FairCC} achieves the required fairness constraint in all clusters without losing much in quality. On increasing the range of plausible fraction of two colors, the total disagreements of \texttt{FairCC} go down but the trends remain similar.

\noindent \textbf{Multiple colors}
Figure~\ref{fig:multiple} compares the performance of \texttt{FairCC} with other baselines. Similar to the case of 2 colors, the quality of \texttt{FairCC} is not much worse than that of \texttt{CC} and is better than any other baseline. This comparison does not plot \texttt{CCMerge} as it does not generate clusters that obey fairness.

\begin{figure}[ht!]
\centering
\includegraphics[scale=0.3]{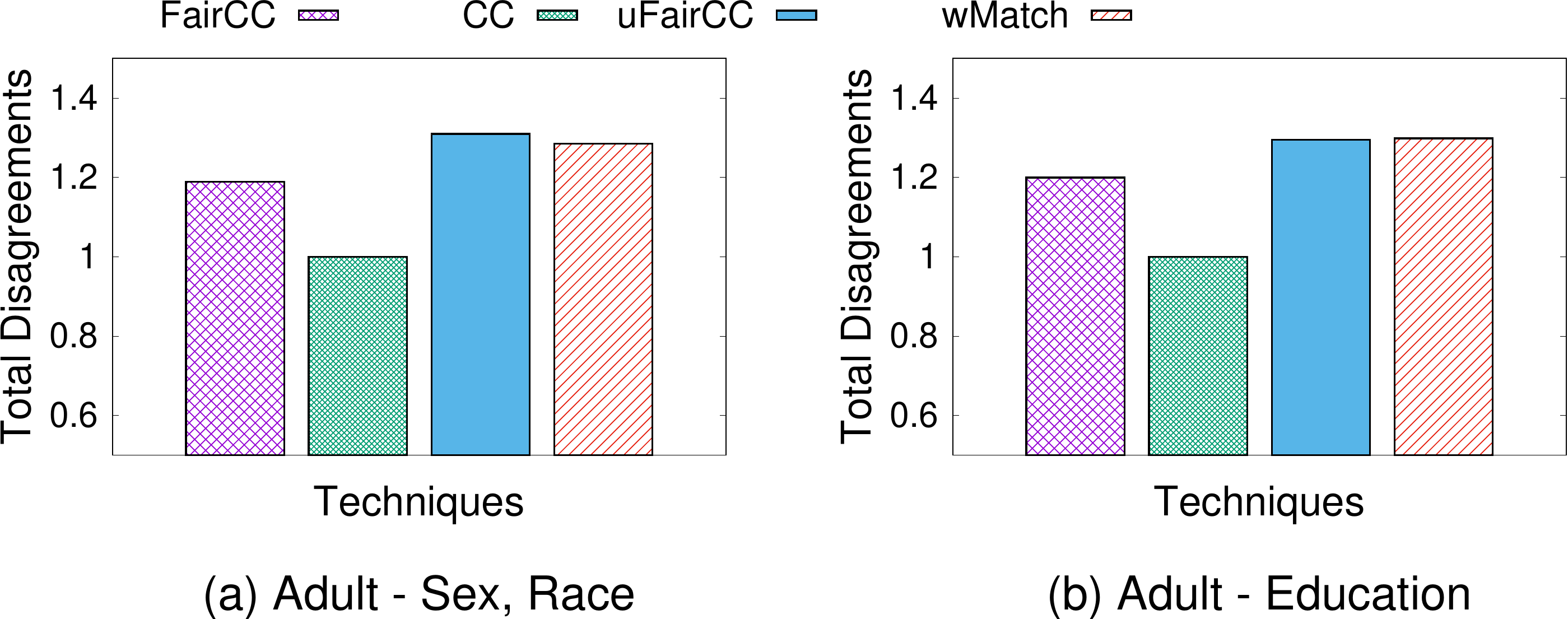}
\caption{Comparison of clusters returned by our algorithm and baselines for instances with more than two colors for Adult dataset. We omit \texttt{CCMerge} as it does not generate fair clusters.\label{fig:multiple}}
\end{figure}

\noindent \textbf{Running Time} \texttt{FairCC} runs in two stages. The first stage identifies a weighted matching between the nodes of different color followed by correlation clustering on one of the color. On all the datasets, our algorithm ran in less than 10 minutes. For a graph of $n$ nodes, with the increase in number of colors the size of subgraph constructed for matching reduces and total running time does not increase.
\section{Conclusions}
In this paper we studied the problem of fairness in clustering. We considered correlation clustering on complete graphs with color constraints to ensure balance and fairness in applications. We obtained combinatorial approximation algorithms for two models. In the first model the goal is to keep distribution of colors in each cluster the same as global distribution while approximately optimizing correlation clustering objective, e.g. minimizing disagreements. In the second model, the goal was to make sure no colors are over-represented or under-represented in each formed cluster while approximately minimizing total number of disagreements. In our experiments we showed our algorithms are effective on real world datasets. Future work could explore extension of our model to general graphs. Obtaining approximation algorithms which achieve better bounds are also of immediate interest.
\bibliography{bibfile}
\bibliographystyle{icml2020}
\appendix
\section{Supplementary Material}
\subsection{Proof of Theorem~\ref{thm:multiple-colors}}

The following lemma could be proved similar to the way Lemma~\ref{lemma:2colors-general} was proved.
\begin{lemma}
\label{lem:multiple-colors-relation-cost-matching-opt}
In each matching $M_i$ constructed in Algorithm~\ref{alg:fairCC-multi-color}, $w(M_i)\leq 2p_i\cdot OPT$.
\end{lemma}

Let $p_{\max} = \max_{i=1}^{|\mathcal{C}|}\{p_i\}$. In the following we show how to pay for all the disagreements within a $\mathcal{O}( (p_{\max})^2\cdot|\mathcal{C}|^2)\cdot OPT$ budget. For simplicity we assume color $c_1$ is red, and there are at least two other colors blue $(c_2)$ and green $(c_3)$. Consider the following cases:

\textbf{Case 1:} This case is similar to Case $1$ in Section~\ref{sec:general-2-colors-CC}. Consider a disagreement between a red vertex (let's say $r_1$), and a node of a different color (let's say blue node $b_3$) such that $r_1$ and $b_3$ are not matched by matching $M_2$. Let's assume $M_2$ matches $b_3$ to $r_2$.

\begin{itemize}
\item{Case 1.1:} If edges $(r_1, b_3)$ and $(r_1, r_2)$ have disagreeing labels, then cost of the edge $(r_2, b_3)$ counted in the $w(M_2)$ is paying for it.

\item{Case 1.2:} If edges $(r_1, b_3)$ and $(r_1, r_2)$ have the same signs, the disagreement on $(r_1, b_3)$ could be charged to the edge $(r_1, r_2)$. The number of such edges charged to $(r_1, r_2)$ is $2p_2$ (and $2p_i$ in general if instead of $b_3$ we considered a node of color $c_i$).
\end{itemize}
\textbf{Case 2:} There exists a disagreement between two non-red nodes from two different hyper nodes, let's say between nodes $b_1$, $g_1$. Let's assume $b_1$ is matched to $r_1$ by $M_2$ (the matching between red and blue nodes), and $g_1$ is matched to $r_2$ by $M_3$ (the matching between red and green nodes).
\begin{itemize}
    \item Case 2.1: Edges $(b_1, g_1)$ and $(r_1, g_1)$ are disagreeing. Then the cost of edge $(r_1, b_1)$ included in the cost of $w(M_2)$ is paying for it.
    \item Case 2.2: Edges $(b_1, g_1)$ and $(r_1, g_1)$ have the same labels and have different labels with $(r_1, r_2)$. We charge the disagreement on $(b_1, g_1)$ and $(r_1, g_1)$ to the edge $(r_2, g_1)$. There are $(|\mathcal{C}|-1)\cdot p_{\max}$ choices for $b_1$ which are all the nodes that are matched to $r_1$ in all the matchings $M_2,\cdots, M_{|\mathcal{C}|}$.
    Therefore in this case, at most $(|\mathcal{C}|-1)\cdot p_{\max}+1$ edges, are charged to the edge  $(r_2, g_1)$, when the matching edge between $(r_2, g_1)$ is paying $1$ for the disagreement between $(r_1, r_2)$ and $(r_1, g_1)$. 
    \item Case 2.3: Edges $(b_1, g_1), (r_1, g_1),(r_1, r_2)$ all have the same labels. We charge disagreements on these edges to $(r_1, r_2)$.
    There are at most $((|\mathcal{C}|-1)\cdot p_{\max})^2$ choices for a pair of non-red nodes like $(b_1, g_1)$, charged to $(r_1, r_2)$.
\end{itemize}
\textbf{Case 3:} This case captures the disagreement between two non-red nodes in the same hyper-node and is similar to Case $3$ in Section~\ref{sec:general-2-colors-CC}. 

There are a total of $((|\mathcal{C}|-1)\cdot p_{\max})^2+2p_{\max}$ charges on edges between red nodes (Cases 1.2 and 2.3) accounting for the total cost to be $(((|\mathcal{C}|-1)\cdot p_{\max})^2+2p_{\max})C$, where $C$ is the correlation clustering objective on red vertices. Similarly, we charge each matched edge $|\mathcal{C}|\cdot p_{\max}+1$ times (Case 2.2) and $p_{\max}-1$ times (Case 3), thereby contributing $\sum_{i=2}^{|\mathcal{C}|}((|\mathcal{C}|+1)\cdot p_{\
max})\cdot w(M_i)$ to the final objective. Considering Lemma~\ref{lem:multiple-colors-relation-cost-matching-opt}, we can conclude the approximation ratio is $\mathcal{O}( (p_{\max})^2\cdot|\mathcal{C}|^2)$, and this completes proof of Theorem~\ref{thm:multiple-colors}.

\noindent \textbf{Note:} When $p_{max}=1$, we can perform classical correlation clustering on nodes of any color $C_i\in \mathcal{C}$ and pick the one that has minimum value. This optimization helps improve the approximation of case 2.3 by reducing the   dependence from $|\mathcal{C}|^2$  to $|\mathcal{C}|$.

\subsection{Proof of Theorem~\ref{thm:over-representation}}
\begin{lemma}
\label{lemma:2colors-over-representation}
$w(M)\leq 2q\cdot OPT$.
\end{lemma}
\begin{proof}
Given the optimum solution $OPT$, we can show a $b$-matching $M'$ could be constructed where endpoints of each matched edge belong to the same cluster, and the degree of each red node is at least $p$ and at most $q$. In the $OPT$ solutin,  in each cluster, the ratio of the number of red to blue nodes is between $1:q$ and $1:p$. Consider a specific cluster $\mathcal{X}$ in the $OPT$ solution, let $n_r, n_b$ denote the number of red and blue nodes in this cluster respectively. Therefore, $n_r\cdot p \leq n_b \leq n_r\cdot q$. Construct a $b$-matching inside $\mathcal{X}$ as following: first assign $p$ distinct blue nodes to each red node in $\mathcal{X}$. If any blue nodes in $\mathcal{X}$ are left un-assigned, assign them to any red node in $\mathcal{X}$ which is assigned to less than $q$ blue nodes. Since $n_b \leq q\cdot n_r$, we can find a $b$-matching with desired properties in each cluster of the $OPT$ solution. $M'$ is the union of the $b$-matchings formed in all the clusters.
First, we show $w(M')\leq 2q\cdot  OPT_G$.
Consider an edge between arbitrary vertices $v_i$ and $v_j$, such that they are not matched in $M'$. If a disagreement occurs on the edge between $(v_i, v_j)$ in $OPT_G$, this disagreement could have been counted at most $2q$ times in $w(M')$. Therefore $w(M')\leq 2q\cdot OPT_G$. Since $M$ is a min cost $b$-matching satisfying degree constraints:
$$w(M)\leq w(M')\leq 2q\cdot OPT_G$$
\end{proof}

The algorithm is as following: run an $\alpha$-approximation correlation clustering on a subset of $G$ which includes the red vertex from each hyper-node (i.e. a collection of matched nodes).
In the following, we show we can pay for all the disagreements within a $\Big((q^2+2q)\cdot\alpha + 4q^2\Big)\cdot OPT$ budget.

\textbf{Case 1:} In Figure~\ref{fig:general-CC}, consider a disagreement between a red vertex  ($r_1$) and a blue ($b_3$) node from different hyper-nodes. 
Two cases might happen:

\begin{itemize}
\item{Case 1.1:} If edges $(r_1, b_3)$ and $(r_1, r_2)$ have disagreeing labels, then cost of the edge $(r_2, b_3)$ counted in the matching  is paying for it.

\item{Case 1.2:} If edges $(r_1, b_3)$ and $(r_1, r_2)$ have the same signs, the disagreement on $(r_1, b_3)$ could be charged to the edge $(r_1, r_2)$. The number of such edges charged to $(r_1, r_2)$ is at most $2q$.
\end{itemize}
\textbf{Case 2:} There exists a disagreement between two blue nodes from two different hyper-nodes, like $(b_1, b_3)$ in Figure \ref{fig:general-CC}.
\begin{itemize}
    \item Case 2.1: Edges $(b_1, b_3)$ and $(r_1, b_3)$ are disagreeing. Then the cost of edge $(r_1, b_1)$ included in the cost of the matching  is paying for it.
    \item Case 2.2: Edges $(b_1, b_3)$ and $(r_1, b_3)$ have the same labels and have different labels with $(r_1, r_2)$. We charge the disagreement on $(b_1, b_3)$ to the edge $(r_2, b_3)$. There are $p$ choices for $b_1$, therefore at most $p$ edges of this type, plus the edge $(r_1, b_3)$ are charged to the edge  $(r_2,b_3)$, when $M$ is paying $1$ for the disagreement between $(r_1, r_2)$ and $(r_1, b_3)$. Therefore, we need to account for at most $q+1$ times the matching cost to account for all edges of this type.
    \item Case 2.3: Edges $(b_1, b_3), (r_1, b_3),(r_1, r_2)$ all have the same labels. There are at most $q^2$ choices for a pair of blue nodes like $(b_1, b_3)$, and disagreements on these edges could be charged to $(r_1, r_2)$.
\end{itemize}
\textbf{Case 3:} A disagreement between two blue nodes in the same hyper-node, $b_1$ and $b_2$ which means $(b_1,b_2)$ is a negative edge. If $(r_1,b_1)$ is positive then $(r_1,b_2)$'s contribution in the matching cost captures it. Similarly, if $(r_1,b_2)$ is a positive edge then the $(r_1,b_2)$'s contribution in matching cost captures this. If both $(r_1,b_1)$ and $(r_1,b_2)$ are negative edges then we can charge both the edges $1/2$. The total number of times an edge $(r_1,b_1)$ is charged is at most $q-1$ as there can be a maximum of $q-1$ negative edges from $b_1$.

There are a total of $q^2+2q$ charges on edges between red nodes (Cases 1.2 and 2.3) accounting for the total cost to be $(q^2+2q)C$, where $C$ is the correlation clustering objective on red vertices. Similarly, we charge each matched edge at most $q+1$ times their weight in Case 2.2 and at most $q-1$ times their weight in Case 3, the total contribution to the final objective is at most $(2q)\cdot w(M)$.  All the charges required to handle cases 1.1 and 2.1 do not add any additional cost to the objective as they are already accounted for the edges considered in $(2q)\cdot w(M)$.
Hence, the total objective value of returned clusters is at most:
\[(q^2+2q)C + (q+1)\cdot w(M) \leq \Big((q^2+2q)\cdot\alpha + 2q\times 2q\Big)\cdot OPT\]
Therefore the approximation ratio is $\mathcal{O}(q^2)$, and this completes proof of Theorem~\ref{thm:over-representation}.

\subsection{Proof of Theorem~\ref{thm:over-representation-multiple-colors}}
Algorithm~\ref{alg:fairCC-multi-color} could be modified to handle this scenario. We need to find a min cost $b$-matching $M_i$ in each iteration where degree of each node of color $c_1$ needs to be between $p_i$ and $q_i$, and degree of each node of color $c_i$ needs to be $1$.
By applying Lemma~\ref{lemma:2colors-over-representation} to each matching $M_i$, one can see $w(M_i) \leq 2q_i\cdot OPT$.

Next, we run an $\alpha$-approximation on the nodes of color $c_1$, and for each fixed vertex $u$ of color $c_1$, all the vertices that are matched to $u$ using any of the matchings $\{M_2,\cdots, M_{|\mathcal{C}|}\}$ go to the same cluster as $u$.

Let $q_{max} = \max\{q_2,\cdots, q_{|\mathcal{C}|}\}$. In the following we show how to pay for all disagreements within a $\mathcal{O}((q_{\max}^2)\cdot |\mathcal{C}|^2)\cdot OPT$ budget. For simplicity let's assume color $c_1$ is red, and there are at least two other colors blue $(c_2)$ and green $(c_3)$. Consider the following cases:

\textbf{Case 1:} Consider a disagreement between a red vertex (let's say $r_1$), and a node of a different color (let's say blue node $b_3$) such that $r_1$ and $b_3$ are not matched by matching $M_2$. Let's assume $M_2$ matches $b_3$ to $r_2$.

\begin{itemize}
\item{Case 1.1:} If edges $(r_1, b_3)$ and $(r_1, r_2)$ have disagreeing labels, then cost of the edge $(r_2, b_3)$ counted in the $w(M_2)$ is paying for it.

\item{Case 1.2:} If edges $(r_1, b_3)$ and $(r_1, r_2)$ have the same signs, the disagreement on $(r_1, b_3)$ could be charged to the edge $(r_1, r_2)$. The number of such edges charged to $(r_1, r_2)$ is at most $2q_2$ (and $2q_i$ in general if instead of $b_3$ we considered a node of color $c_i$).
\end{itemize}
\textbf{Case 2:} There exists a disagreement between two non-red nodes from two different hyper nodes, let's say between nodes $b_1$, $g_1$. Let's assume $b_1$ is matched to $r_1$ by $M_2$ (the matching between red and blue nodes), and $g_1$ is matched to $r_2$ by $M_3$ (the matching between red and green nodes).
\begin{itemize}
    \item Case 2.1: Edges $(b_1, g_1)$ and $(r_1, g_1)$ are disagreeing. Then the cost of edge $(r_1, b_1)$ included in the cost of $w(M_2)$ is paying for it.
    \item Case 2.2: Edges $(b_1, g_1)$ and $(r_1, g_1)$ have the same labels and have different labels with $(r_1, r_2)$. We charge the disagreement on $(b_1, g_1)$ and $(r_1, g_1)$ to the edge $(r_2, g_1)$. There are $(|\mathcal{C}|-1)\cdot q_{\max}$ choices for $b_1$ which are all the nodes that are matched to $r_1$ in all the matchings $M_2,\cdots, M_{|\mathcal{C}|}$.
    Therefore in this case, at most $(|\mathcal{C}|-1)\cdot q_{\max}+1$ edges, are charged to the edge  $(r_2, g_1)$, when the matching edge between $(r_2, g_1)$ is paying $1$ for the disagreement between $(r_1, r_2)$ and $(r_1, g_1)$. 
    \item Case 2.3: Edges $(b_1, g_1), (r_1, g_1),(r_1, r_2)$ all have the same labels. We charge disagreements on these edges to $(r_1, r_2)$.
    There are at most $((|\mathcal{C}|-1)\cdot q_{\max})^2$ choices for a pair of non-red nodes like $(b_1, g_1)$, charged to $(r_1, r_2)$.
\end{itemize}
\textbf{Case 3:} This case captures the disagreement between two non-red nodes in the same hyper-node and is similar to Case $3$ in Section~\ref{sec:general-2-colors-CC}. 

There are a total of $((|\mathcal{C}|-1)\cdot q_{\max})^2+2q_{\max}$ charges on edges between red nodes (Cases 1.2 and 2.3) accounting for the total cost to be $(((|\mathcal{C}|-1)\cdot q_{\max})^2+2q_{\max})C$, where $C$ is the correlation clustering objective on red vertices. Similarly, we charge each matched edge $|\mathcal{C}|\cdot q_{\max}+1$ times (Case 2.2) and $q_{\max}-1$ times (Case 3), thereby contributing $\sum_{i=2}^{|\mathcal{C}|}((|\mathcal{C}|+1)\cdot q_{\
max})\cdot w(M_i)$ to the final objective. Considering Lemma~\ref{lem:multiple-colors-relation-cost-matching-opt}, we can conclude the approximation ratio is $\mathcal{O}( (q_{\max})^2\cdot|\mathcal{C}|^2)$, and this completes proof of Theorem~\ref{thm:multiple-colors}.

\end{document}